\documentclass[12pt]{article}
\usepackage{graphicx}

\def\hybrid{\topmargin 0pt      \oddsidemargin 0pt
        \headheight 0pt \headsep 0pt
       \voffset-1cm
        \textwidth 6.25in       
       \textheight 9.5in       
        \marginparwidth 0.0in
        \parskip 5pt plus 1pt   \jot = 1.5ex}
\catcode`\@=11
\def\marginnote#1{}

\newcount\hour
\newcount\minute
\newtoks\amorpm
\hour=\time\divide\hour by60
\minute=\time{\multiply\hour by60 \global\advance\minute by-\hour}
\edef\standardtime{{\ifnum\hour<12 \global\amorpm={am}%
        \else\global\amorpm={pm}\advance\hour by-12 \fi
        \ifnum\hour=0 \hour=12 \fi
        \number\hour:\ifnum\minute<10 0\fi\number\minute\the\amorpm}}
\edef\militarytime{\number\hour:\ifnum\minute<10 0\fi\number\minute}

\def\draftlabel#1{{\@bsphack\if@filesw {\let\thepage\relax
   \xdef\@gtempa{\write\@auxout{\string
      \newlabel{#1}{{\@currentlabel}{\thepage}}}}}\@gtempa
   \if@nobreak \ifvmode\nobreak\fi\fi\fi\@esphack}
        \gdef\@eqnlabel{#1}}
\def\@eqnlabel{}
\def\@vacuum{}
\def\draftmarginnote#1{\marginpar{\raggedright\scriptsize\tt#1}}

\def\draftlabel#1{{\@bsphack\if@filesw {\let\thepage\relax
   \xdef\@gtempa{\write\@auxout{\string
      \newlabel{#1}{{\@currentlabel}{\thepage}}}}}\@gtempa
   \if@nobreak \ifvmode\nobreak\fi\fi\fi\@esphack}
        \gdef\@eqnlabel{#1}}
\def\@eqnlabel{}
\def\@vacuum{}
\def\draftmarginnote#1{\marginpar{\raggedright\scriptsize\tt#1}}

\def\draft{\oddsidemargin -.5truein
        \def\@oddfoot{\sl preliminary draft \hfil
        \rm\thepage\hfil\sl\today\quad\militarytime}
        \let\@evenfoot\@oddfoot \overfullrule 3pt
        \let\label=\draftlabel
        \let\marginnote=\draftmarginnote
   \def\@eqnnum{(\theequation)\rlap{\kern\marginparsep\tt\@eqnlabel}%
\global\let\@eqnlabel\@vacuum}  }

\def\numberbysection{\@addtoreset{equation}{section}
        \def\theequation{\thesection.\arabic{equation}}}

\def\underline#1{\relax\ifmmode\@@underline#1\else
        $\@@underline{\hbox{#1}}$\relax\fi}

\def\titlepage{\@restonecolfalse\if@twocolumn\@restonecoltrue\onecolumn
     \else \newpage \fi \thispagestyle{empty}\c@page\z@
        \def\thefootnote{\fnsymbol{footnote}} }

\def\endtitlepage{\if@restonecol\twocolumn \else  \fi
        \def\thefootnote{\arabic{footnote}}
        \setcounter{footnote}{0}}  
\relax

\hybrid

\newfont{\Bbb}{msbm10 scaled 1\@ptsize00}
\newfont{\Bbbb}{msbm7 scaled 1\@ptsize00}

\newcommand{\DDD}{\raise-1pt\hbox{$\mbox{\Bbbb D}$}}

\newcommand{\UUU}{\raise-1pt\hbox{$\mbox{\Bbbb U}$}}

\newcommand{\ZZ}{\mbox{\Bbb Z}}
\newcommand{\z}{\raise-1pt\hbox{$\mbox{\Bbbb Z}$}}

\newcommand{\sss}{\raise-1pt\hbox{$\mbox{\Bbbb S}$}}

\def\beq{\begin{equation}}
\def\eeq{\end{equation}}
\def\p{\partial}

\newtheorem{theorem}{Theorem}

\newtheorem{lemma-definition}{Lemma-Definition}

\newtheorem{remark}{Remark}

\newtheorem{proposition}{Proposition}

\def\square{\hfill
{\vrule height6pt width6pt depth1pt} \break \vspace{.01cm}}

\begin{document}

\begin{titlepage}

\title{Elliptic solutions to matrix CKP equation}

\author{A. Zabrodin\thanks{National Research University 
Higher School of Economics,
20 Myasnitskaya Ulitsa, Moscow 101000, Russia, 
and NRC ``Kurchatov institute'', Moscow, Russia;
e-mail: zabrodin@itep.ru}}

\date{August 2026}

\maketitle

\vspace{-7cm} \centerline{ \hfill ITEP-TH-31/26}\vspace{7cm}

\begin{abstract}

A class of elliptic solutions to the matrix
CKP equation is studied. Equations of motion for
their poles and matrix coefficients at the poles
are obtained. As in the scalar case, they are 
of the first order, in contrast to what takes
place in the KP and BKP hierarchies, where the equations
of motion are of the second order.

\end{abstract}

\end{titlepage}

\vspace{5mm}

%

\tableofcontents

\vspace{5mm}

\section{Introduction}

The study of singular solutions to nonlinear integrable
equations and dynamics of their poles $x_i$ in the
space variable $x$
was initiated in the pioneering papers
\cite{AMM77}--\cite{CC77}. Now it 
is a well known subject in the theory of integrable systems. 
The remarkable result is
that poles of solutions to the Kadomtsev-Petviashvili (KP) equation, 
as functions of a time variable move 
as particles of the integrable Calogero-Moser many-body system
\cite{Calogero71}--\cite{OP81}. Elliptic, 
trigonometric and rational solutions
correspond respectively to the Calogero-Moser systems with 
elliptic, trigonometric and rational potential functions. 
Later the correspondence with many-body integrable systems
was extended to other equations and
their hierarchies (see \cite{Shiota94}--\cite{PZ23}).

As is known, hierarchies of nonlinear integrable equations
of the KP type
admit integrable multi-component and matrix 
generalizations \cite{DJKM81a}--\cite{Teo11}.
Singular solutions to the matrix KP
equation were investigated in \cite{KBBT95,PZ18,PZ21spin}. 
In the matrix case, the data
of singular solutions include not only positions 
of poles $x_i$ but also
some ``internal degrees of freedom'' 
which are matrix coefficients at the poles
(they are fixed in the scalar case). 
In the work \cite{KBBT95} it was shown that 
the dynamics of the data of such solutions
with respect to the time $t_2$ of the matrix
KP hierarchy is isomorphic to the dynamics of a 
spin generalization of the Calogero-Moser 
system which is also known as the Gibbons-Hermsen model \cite{GH84}.

In the paper \cite{DJKM81} an infinite integrable hierarchy of partial differential
equations with ${\rm Sp}\, (\infty)$ symmetry was introduced. 
It is called the
Kadomtsev-Petviashvili hierarchy of type C (CKP).
It was further studied in \cite{DM-H09}--\cite{KZ21}.
In \cite{KL23,Z24} a multi-component generalization of this hierarchy
was suggested.
In particular, its elliptic solutions were investigated in
\cite{KZ21}, where it was shown that poles $x_i$, as functions of time,
satisfy differential equations of {\it first order}, in contrast
to dynamical equations for poles of singular solutions to the
KP hierarchy, which are of second order.

The aim of this paper is to extend this result to the matrix
CKP equation. As in our earlier works on dynamics of poles,
we use the Krichever method \cite{Krichever78,Krichever80}.
In this approach, it is not necessary to know an explicit form
of the nonlinear equations in question 
(which often may be too complicated). Instead, one deals with
auxiliary linear problems for the Baker-Akhiezer function 
whose compatibility leads to the
nonlinear equation, then separation of variables is clear from
the very beginning. In accordance with this approach, we 
consider the first nontrivial linear problem for the matrix CKP
equation, which allows us to derive equations of motion
for poles and the corresponding matrix coefficients 
in $t_3$ without any use of the nonlinear equation itself.

In a nutshell, the main result of this work 
is as follows. The Lax operator
for the matrix CKP hierarchy is the pseudodifferential operator
with matrix coefficients
\beq\label{multi3a}
{\cal L}=\p_x +u_1\p_x^{-1}+u_2 \p_x^{-2}+\ldots \, 
\eeq
satisfying the constraint ${\cal L}^{\dag}=-{\cal L}$ (the conjugation
$(\ldots )^{\dag}$ is defined in the standard way, see Section 2).
For elliptic solutions,
the functions $u_1(x)$, $u_2(x)$, etc are double-periodic in the space
variable $x$.
In particular, matrix elements of the 
function $u_1$ are of the form
\beq\label{el7a}
u_{1, \alpha \beta}=-\sum_{i}a_i^{\alpha}a_i^{\beta}\wp (x-x_i),
\eeq
where $\wp (x)$ is the Weierstrass $\wp$-function and 
$a_i^{\alpha}$ are components of some time-dependent vectors
${\bf a}_i$
($\alpha =1, \ldots , n$).

To find the dependence of $x_i$ and $a_i^{\alpha}$ on $t_3$
one should consider the first nontrivial linear problem for the matrix 
Baker-Akhiezer function $\Psi$:
\beq\label{psi}
\p_{t_3}\Psi = B_3 \Psi ,
\eeq
where $B_3=({\cal L}^3)_{\geq 0}$ is the differential 
operator of the form
\beq\label{B1a}
B_3=\p_x^3 +3u_1 \p_x +3(u_1'+u_2)
\eeq
(the prime means the $x$-derivative). Since the coefficients
are double-periodic functions,
it is natural to find solutions for $\Psi$ 
among double-Bloch functions.

The main theorem proved in Section 4 is

\begin{theorem} \label{theorem:main1}
Let a solution to the matrix CKP equation
$u_1$ be of the form (\ref{el7a}). 
Then $x_i$ and $a_i^{\alpha}$ satisfy the following 
equations\footnote{Here and below summation over repeated Greek 
indices is implied.}:
\beq\label{e221}
\begin{array}{l}
\displaystyle{
\dot x_i=3\sum_{j\neq i}(a_i^{\gamma}a_j^{\gamma})^2 \wp (x_i-x_j),}
\\ \\
\displaystyle{
\dot a_i^{\alpha}=3\sum_{j\neq i}\sum_{k\neq j} (a_j^{\nu}a_k^{\nu})
\Bigl ((a_i^{\gamma}a_k^{\gamma})a_j^{\alpha}\! -\! (a_i^{\gamma}a_j^{\gamma})a_k^{\alpha}\Bigr ) \wp (x_j-x_k)\zeta (x_i-x_j)}
\\ \\
\displaystyle{\phantom{aaaaaaaaaaaaa}
+\, \frac{3}{2}\sum_{j\neq i} (a_i^{\gamma}a_j^{\gamma})
\Bigl ( a_j^{\alpha}-(a_i^{\nu}a_j^{\nu})a_i^{\alpha}\Bigr )
\wp '(x_i-x_j)},
\end{array}
\eeq
where dot means the $t_3$-derivative and
$\zeta (x)$ is the Weierstrass $\zeta$-function.
\end{theorem}

The organization of the paper is as follows.
In Section 2 we start with a more general case of the 
multi-component CKP hierarchy, in which the matrix
CKP equation can be naturally embedded. This material is not new, see
\cite{KL23,Z24}. Section 3 is devoted to the restriction to the
matrix hierarchy, which is a special subhierarchy of the 
multi-component one. The necessary details on the operator 
$B_3$ are given. Section 4 contains derivation of the main
result, equations of motion (\ref{e221}). It
occupies a central place in the article. Section 5 is a conclusion.
The appendix contains the definition and the main properties of
the Weierstrass functions and the Lam\'e-Hermite function.

\section{The multi-component CKP hierarchy: the Lax-Sato
formalism}

In this paper we work within the Lax-Sato formalism.
For the equivalent bilinear formulation see \cite{Z24}.

The set of independent variables (``times'') of the hierarchy is
\beq\label{multi1}
{\bf t}=\{{\bf t}_1, {\bf t}_2, \ldots , {\bf t}_n\}, \quad
{\bf t}_{\alpha}=\{t_{\alpha , 1}, t_{\alpha , 3}, t_{\alpha , 5},
\ldots \}, \quad \alpha =1, \ldots , n.
\eeq
It is divided into $n$ infinite sets, and the variables in each 
set are indexed by positive odd numbers.
It is convenient to introduce the variable $x$ such that
\beq\label{multi2}
\p_x =\sum_{\alpha =1}^n \p_{t_{\alpha , 1}}.
\eeq
Let $f({\bf t})$ be any function of ${\bf t}$. It may be regarded
also as a function of $x$ as follows:
\beq\label{f}
f(x, {\bf t})=f(t_{1,1}+x, t_{2,1}+x, \ldots , t_{n,1}+x;
t_{1,3}, t_{2,3},\ldots , t_{n,3}; \ldots ).
\eeq

In the framework of the Lax-Sato formalism,
the hierarchy is an infinite set of evolution equations 
in the times ${\bf t}$ 
for matrix functions of the variables $x, {\bf t}$.
The main object is the Lax operator which is a 
pseudo-differential operator of the form
\beq\label{multi3}
{\cal L}=\p_x +u_1\p_x^{-1}+u_2 \p_x^{-2}+\ldots \, ,
\eeq
where the coefficients $u_i=u_i(x)$ are $n\! \times \! n$ matrices
with the constraint
\beq\label{multi4}
{\cal L}^{\dag}=-{\cal L}.
\eeq
Here $\dag$ means the formal adjoint
defined by the rule 
$\Bigl (f(x)\circ \p_x^{n}\Bigr )^{\dag}=(-\p_x)^n \circ f^{\rm T}(x)$
and $f^{\rm T}$ is the transposed matrix $f$.
The coefficient functions $u_k$ depend on 
$x$ and also on all the times:
$$
u_k(x, {\bf t})=u_k(x+t_{1,1}, x+t_{2,1}, \ldots , x+t_{n,1};
t_{1,3}, \ldots , t_{n,3}; \ldots ).
$$ 
Besides, there are
other matrix pseudo-differential operators 
${\cal R}_1, \ldots , {\cal R}_n$
of the form
\beq\label{multi5}
{\cal R}_{\alpha}=E_{\alpha}+ u_{\alpha , 1}
\p_x^{-1}+u_{\alpha , 2}\p_x^{-2}+\ldots ,
\eeq
where $E_{\alpha}$ is the $n \! \times \! n$ 
matrix with the $(\alpha , \alpha )$ element
equal to 1 and all other components equal to $0$, 
and $u_{\alpha , i}$ are 
$n \! \times \! n$ matrices. 
The operators ${\cal L}$, ${\cal R}_1, \ldots , {\cal R}_n$
satisfy the conditions
\beq\label{multi6}
{\cal L}{\cal R}_{\alpha}={\cal R}_{\alpha}{\cal L}, \quad
{\cal R}_{\alpha}{\cal R}_{\beta}=\delta_{\alpha \beta}{\cal R}_{\alpha}, \quad
\sum_{\alpha =1}^n {\cal R}_{\alpha}=I,
\eeq
where $I$ is the unity $n\times n$ matrix.
In addition, the operators
${\cal R}_{\alpha}$ are required to satisfy the constraints
\beq\label{multi7}
{\cal R}^{\dag}_{\alpha}={\cal R}_{\alpha}.
\eeq

The Lax equations which define evolution in the times read
\beq\label{multi8}
\p_{t_{\alpha , k}}{\cal L}=[B_{\alpha , k}, \, {\cal L}], \quad
\p_{t_{\alpha , k}}{\cal R}_{\beta}=[B_{\alpha , k}, \, {\cal R}_{\beta}],
\quad B_{\alpha , k} = \Bigl ({\cal L}^k{\cal R}_{\alpha}\Bigr )_{\geq 0},
\quad k=1,3,5, \ldots ,
\eeq
where $(\ldots )_{\geq 0}$ means 
the differential part of a pseudo-differential operator.

\begin{proposition}
The constraints (\ref{multi4}), (\ref{multi7}) are preserved
by the evolution defined by (\ref{multi8}).
\end{proposition}

\noindent
{\it Proof.}
Indeed, for any odd $k$ we have
\beq\label{multi9}
\p_{t_{\alpha , k}}({\cal L}+{\cal L}^{\dag})=[B_{\alpha , k}, 
\, {\cal L}]-
[B^{\dag}_{\alpha , k}, \, {\cal L}^{\dag}]=0
\eeq
and
\beq\label{multi10}
\p_{t_{\alpha , k}}({\cal R}_{\beta}-{\cal R}_{\beta}^{\dag})=
[B_{\alpha , k}, \, {\cal R}_{\beta}]+[B^{\dag}_{\alpha , k}, \, {\cal R}^{\dag}_{\beta}]=0
\eeq
since $B_{\alpha, k}^{\dag}=-B_{\alpha, k}$ for odd $k$.
\square

An important object of the theory is the ``dressing'' operator, which
is the matrix pseudo-differential operator of the form
\beq\label{multi11}
{\cal W}=I+\xi_1 \p_x^{-1}+\xi_2 \p_x^{-2}+\ldots ,
\eeq
where $\xi_i$ are $n\! \times \! n$ matrix-valued functions of $x$.
Using the dressing operator, we can express
${\cal L}$ and ${\cal R}_{\alpha}$ by ``dressing'' the bare operators
$\p_x$ and $E_{\alpha}$: 
\beq\label{multi12}
{\cal L}={\cal W}\p_x {\cal W}^{-1}, \qquad
{\cal R}_{\alpha}={\cal W}E_{\alpha} {\cal W}^{-1},
\eeq
then the constraints (\ref{multi6}) hold
identically by construction (since the bare operators
$\p_x$, $E_{\alpha}$ trivially satisfy them).  
However, there is an ambiguity in the definition of the dressing 
operator: it can be multiplied
from the right by any pseudo-differential operator with constant 
coefficients
commuting with $E_{\alpha}$ for any $\alpha$.
The constraints (\ref{multi4}), (\ref{multi7}) imply that
${\cal W}^{\dag}{\cal W}$ commutes with $\p_x$ and $E_{\alpha}$, i.e.,
it is a pseudo-differential operator with constant coefficients 
commuting with $E_{\alpha}$.
The ambiguity in the definition of the dressing operator can be fixed by 
imposing the condition
\beq\label{multi13}
{\cal W}^{\dag}={\cal W}^{-1}.
\eeq

With the help of the dressing operator, one can introduce
another useful object, the 
matrix Baker-Akhiezer (BA) function:
\beq\label{multi14}
\Psi (x,{\bf t}, z)={\cal W} \exp \Bigl 
( xzI +\sum_{\alpha =1}^n E_{\alpha}
\xi ({\bf t}_{\alpha}, z)\Bigr ).
\eeq
Here
\beq\label{multi15}
\xi ({\bf t}_{\alpha}, z)=\sum_{k\geq 1, \,\, {\rm odd}}t_{\alpha , k}z^k.
\eeq
The BA function has the expansion
\beq\label{multi16}
\Psi_{\alpha \beta} (x,{\bf t}, z)=\Biggl (\delta_{\alpha \beta}+
\sum_{k\geq 1} \xi_{k, \alpha \beta}(x, {\bf t})z^{-k}\Biggr )
e^{xz+ \xi ({\bf t}_{\beta}, z)}.
\eeq
as $z\to \infty$ and satisfies the linear equations
\beq\label{multi17}
{\cal L}\Psi =z\Psi , \qquad
\p_{t_{\alpha , k}}\Psi =B_{\alpha , k}\Psi \quad (\mbox{$k$ odd}).
\eeq

\begin{remark}\label{remark:embedding}
As it is seen from the definition, the multi-component CKP hierarchy
is a subhierarchy of the multi-component KP one. The precise 
characterization of those KP-solutions that give rise to CKP-solutions
(after putting the ``even'' times $t_{\alpha , 2k}$ equal to zero)
is given in \cite{KZ21,Z24}.
For our purposes below, it will be useful to think of a
CKP-solution as being embedded into the KP hierarchy in this way.
So, dependent variables of the CKP (for example, the $\xi_i$'s),
which were originally functions of the ``odd'' times only, can be
regarded as functions of the ``even'' times as well. In particular,
we will need the derivative $\p_{t_2}\xi_1$ taken at $t_2=0$.
\end{remark}

\section{The matrix CKP hierarchy}

The matrix CKP hierarchy is a subhierarchy of the 
multi-component CKP one. It is 
obtained from it the after restricting
the independent variables in the following manner:
\beq\label{multi18}
t_{\alpha , m}=t_m \quad \mbox{for each odd $m$ and
all $\alpha =1, \ldots , n$}.
\eeq
The corresponding vector fields are related as $\displaystyle{
\p_{t_m}=\sum_{\alpha =1}^n \p_{t_{\alpha , m}}}$. 
Accordingly, we have the expansion
\beq\label{multi16a}
\Psi =\Bigl (I+
\sum_{k\geq 1} \xi_{k}z^{-k}\Bigr )
\exp \Bigl (xz+\!\!\! \sum_{k\geq 1, \,\, {\rm odd}}t_k z^{k}\Bigr )
\eeq
of the BA function and the matrix linear equations
\beq\label{multi17a}
\p_{t_{k}}\Psi =B_{k}\Psi , \quad B_k =\Bigl 
({\cal L}^k\Bigr )_{\geq 0} \quad (\mbox{$k$ odd})
\eeq
for it. Obviously, $B_1=\p_x$, and so we can identify the variable
$x$ with $t_1$. 

The following statement is well known.

\begin{proposition}
The compatibility conditions for 
equations (\ref{multi17a}) have the form of the
Zakharov-Shabat equations
\beq\label{multi17b}
\p_{t_l}B_k-\p_{t_k}B_l+[B_k, B_l]=0, \qquad \mbox{$k,l$ odd}.
\eeq
\end{proposition}
\square

\noindent
The matrix CKP equation (or rather a system of 
equations) is obtained from (\ref{multi17b}) 
at $k=3$, $l=5$. We do not need its explicit form because the
Krichever method implies separation of the
variables $t_3$ and $t_5$ from the very beginning. So, all what
we need is the explicit form of the operator $B_3$.
It is given by the 
following proposition. 

\begin{proposition}
The operator $B_3$ has the form
\beq\label{B4a}
\begin{array}{c}
B_3=\p_x^3 -3\xi_1'\p_x -\frac{3}{2}\, \xi_1'' -\frac{3}{2}\, 
\p_{t_2}\xi_1,
\end{array}
\eeq
where the derivative with respect to $t_2$ is taken at $t_2=0$, as it is
explained in Remark \ref{remark:embedding}.
\end{proposition}

\noindent
{\it Proof.}
In terms of the coefficient functions of the
Lax operator we have:
\beq\label{B1}
B_3=\p_x^3 +3u_1 \p_x +3(u_1'+u_2),
\eeq
where prime means the $x$-derivative and the constraint
(\ref{multi4}) implies that
$$u_1^{\rm T}=u_1, \quad u_2+u_2^{\rm T}=-u_1'.$$
Using (\ref{multi12}), one can also express 
the operator $B_3$ in terms of the
coefficients $\xi_i$ of the dressing operator. 
From the condition (\ref{multi13})
it follows that
\beq\label{B2}
\xi_1^{\rm T}=\xi_1, \qquad \xi_2+\xi_2^{\rm T}+\xi_1'-\xi_1^2=0
\eeq
and from (\ref{multi12}) we have
\beq\label{B3}
\begin{array}{l}
u_1=-\xi_1',
\\ \\
u_2=\frac{1}{2}\, \xi_1'' +\frac{1}{2}
\Bigl ([\xi_1', \xi_1]+{\xi_2'}^{\rm T }-\xi_2'\Bigr ).
\end{array}
\eeq
Therefore, the operator $B_3$ reads
\beq\label{B4}
\begin{array}{c}
B_3=\p_x^3 -3\xi_1'\p_x -\frac{3}{2}\, \xi_1'' +\frac{3}{2}\, w,
\end{array}
\eeq
where the antisymmetric matrix $w$ is
\beq\label{B5}
w=[\xi_1', \xi_1]+{\xi_{2}'}^{\rm T }-\xi_2'.
\eeq
One can directly check that $B_3^{\dag}=-B_3$.
In order to represent $w$ in the desired form, 
we use the embedding of the CKP
hierarchy into the KP, as explained in Remark \ref{remark:embedding},
and consider the auxiliary linear problem
for the KP $t_2$-flow at $t_2=0$:
\beq\label{B6}
\p_{t_2}\Psi = \p_x^2\Psi +2u_1 \Psi .
\eeq
With the expansion (\ref{multi16}) and the first relation in
(\ref{B3}) this equation implies that
\beq\label{B7}
\p_{t_2}\xi_1 =2\xi_2' -2\xi_1' \xi_1 +\xi_1''.
\eeq
Summing with the transposed equation and using the second relation in (\ref{B2}), we
conclude that $\p_{t_2}\xi_1$ at $t_2=0$ is an antisymmetic matrix:
$(\p_{t_2}\xi_1)^{\rm T}=-\p_{t_2}\xi_1$. Therefore, 
subtracting (\ref{B7}) from the
transposed equation, we get
\beq\label{B8}
w=[\xi_1', \xi_1]+\xi_2'{}^{\rm T}-\xi_2' =-\p_{t_2}\xi_1,
\eeq
which yields the required expression (\ref{B4a}).
\square

\section{Elliptic solutions of the matrix CKP equation}

By elliptic solutions of the CKP equation we mean solutions such that
the coefficient functions $u_i$ of the Lax operator are elliptic (double-periodic
in the complex plane) functions of the variable $x$, i.e. there exist two
complex numbers $\omega_1$, $\omega_2$ with ${\rm Im}(\omega_2/ \omega_1)>0$ such that
$u_i(x+2\omega_a)=u_i(x)$, $a=1,2$.
From (\ref{B3}) we see that a solution
is elliptic if $\xi_1'$ and $[\xi_1', \xi_1]+{\xi_{2}'}^{\rm T }-\xi_2'$
are elliptic functions of $x$
with periods $2\omega_a$.

As soon as coefficients of the equation
\beq\label{el1}
\p_{t_3}\Psi =\p_x^3\Psi +3u_1\p_x\Psi +3(u_1'+u_2)\Psi
\eeq
are double-periodic functions of 
$x$, we can look for double-Bloch solutions, i.e.
solutions such that $\Psi (x+2\omega_a)=b_a\Psi (x)$ with some Bloch multipliers $b_a$.

As a building block for double-Bloch functions,
we will use the Lam\'e-Hermite function $\Phi (x, \lambda )$
defined by equation (\ref{Phi}) in the appendix.
We will often suppress the second argument of 
this function writing simply
$\Phi (x)$ instead of $\Phi (x, \lambda )$.
We will also need the $x$-derivatives
$\Phi '(x, \lambda )=\p_x \Phi (x, \lambda )$, 
$\Phi ''(x, \lambda )=\p^2_x \Phi (x, \lambda )$
and so on.

Any non-trivial double-Bloch function (which is not just 
the exponential function)
must have poles in the fundamental domain. Therefore, we have the following
representation of the matrix wave function:
\beq\label{el5}
\Psi_{\alpha \beta}=\exp \Bigl (xz +\!
\sum_{k\geq 1,\, \, {\rm odd}}t_kz^k\Bigr )\sum_{i=1}^n \rho_{i, \alpha  \beta}
\Phi (x-x_i, \lambda ).
\eeq
As it follows from (\ref{quasi}), 
it is indeed a double-Bloch function with Bloch
multipliers
\beq\label{Bloch}
b_{a}=e^{2\omega_{a} (z-\zeta (\lambda )) + 2\zeta (\omega_{a} )\lambda }.
\eeq
The parameters $z, \lambda$ are spectral parameters.
They are connected by the equation of the spectral curve given below.

In \cite{KBBT95,PZ18} it was shown that the residue
at any pole $x_i$ in $x$ of the matrix wave function of the matrix KP hierarchy
is a rank one matrix. Since the CKP 
wave function is the wave function for a solution
of the KP hierarchy with ``even'' times
$t_{2k}$ put equal to 0, we can write
\beq\label{el5a}
\rho_{i, \alpha \beta}=a_i^{\alpha}c_i^{\beta},
\eeq
where $a_i^\alpha$, $c_i^{\alpha}$ are components of some vectors
${\bf a}_i=(a_i^1, \ldots , a_i^n)^T$, 
${\bf c}_i=(c_i^1, \ldots , c_i^n)^T$.
The vectors ${\bf a}_i$ depend on 
the times $t_3, \, t_5, \ldots$ while the vectors
${\bf c}_i$ depend on the same set of times and on $z$.
As it is shown in \cite{KBBT95},
the function (\ref{el5}) is essentially the Baker-Akhiezer function on the spectral curve. (However,
it differs from the wave function 
(\ref{multi16a}) by a normalization factor.)
Since we are interested in dynamics in $t_3$, all other times 
will be put equal to 0 in what follows. So, the ansatz for the
wave function is
\beq\label{el5b}
\Psi_{\alpha \beta}=e^{xz +t_3z^3}
\sum_{i=1}^n a_i^{\alpha}c_i^{\beta}
\Phi (x-x_i, \lambda ).
\eeq

Let us recall the story about elliptic solutions to the 
matrix KP equation \cite{KBBT95}.

\begin{theorem} (\cite{KBBT95}) \label{theorem:KP}
Let $\Psi$ 
be a double-Bloch solution to equation (\ref{B6})
with simple poles at some points $x_i$, $i=1, \ldots , N$. Then the
matrix elements of the coefficient $\xi_1$ have the form
\beq\label{kp1}
\xi_{1, \alpha \beta}=S_{\alpha \beta}-
\sum_{i=1}^N a_i^{\alpha}b_i^{\beta}\zeta (x-x_i),
\eeq
where $S_{\alpha \beta}$ is a constant matrix and
$a_i^\alpha$, $b_i^{\beta}$ are components of some vectors
${\bf a}_i$, ${\bf b}_i$ such that
\beq\label{el9a}
a_i^{\gamma}b_i^{\gamma}=1 \quad 
\mbox{for all $\, i=1, \ldots , N$}.
\eeq
They satisfy the following equations:
\beq\label{kp2}
\p_{t_2}a_i^{\alpha}=-2\sum_{k\neq i}a_k^{\alpha}
(a_i^{\gamma}b_k^{\gamma})\wp (x_i-x_k),
\qquad
\p_{t_2}b_i^{\beta}=2\sum_{k\neq i}b_k^{\beta}(a_k^{\gamma}b_i^{\gamma})
\wp (x_i-x_k).
\eeq
Dynamics of the poles is described by the equations
\beq\label{kp2a}
\p_{t_2}^2 x_i=4\sum_{j\neq i}(b_i^{\mu} a_j^{\mu} ) (b_j^{\nu}
 a_i^{\nu}) \wp '(x_i-x_j).
\eeq
Together with (\ref{kp2}), they are equations of motion of the
spin generalization of the Ca\-lo\-ge\-ro-Moser system.
\end{theorem}

The main result of this paper is the following CKP-analogue
of Theorem \ref{theorem:KP}.

\begin{theorem} \label{theorem:main}
Let $\Psi$ 
be a double-Bloch solution to equation (\ref{el1})
with simple poles at some points $x_i$, $i=1, \ldots , N$. Then the
matrix elements of the coefficient $\xi_1$ have the form
\beq\label{kp1a}
\xi_{1, \alpha \beta}=S_{\alpha \beta}-
\sum_{i=1}^N a_i^{\alpha}a_i^{\beta}\zeta (x-x_i),
\eeq
where $S_{\alpha \beta}$ is a constant matrix and
$a_i^\alpha$ are components of some vectors ${\bf a}_i$ such that
\beq\label{el9}
a_i^{\gamma}a_i^{\gamma}=1.
\eeq
They satisfy the following equations:
\beq\label{e24}
\begin{array}{l}
\dot a_i^{\alpha} =\displaystyle{
\frac{3}{2}\sum_{j\neq i} (a_i^{\gamma}a_j^{\gamma})
\Bigl ( a_j^{\alpha}-(a_i^{\nu}a_j^{\nu})a_i^{\alpha}\Bigr )
\wp '(x_i-x_j)}
\\ \\
\phantom{aaaaaaaaaaa}\displaystyle{
+\,\, 3\sum_{j\neq i}\sum_{k\neq j} (a_j^{\nu}a_k^{\nu})
\Bigl ((a_i^{\gamma}a_k^{\gamma})a_j^{\alpha}\! -\! (a_i^{\gamma}a_j^{\gamma})a_k^{\alpha}\Bigr )
\wp (x_j-x_k)\zeta (x_i-x_j)},
\end{array}
\eeq
where dot means $t_3$-derivative.
Dynamics of the poles is described by the equations
\beq\label{e22}
\dot x_i=3\sum_{j\neq i}(a_i^{\gamma}a_j^{\gamma})^2 \wp (x_i-x_j).
\eeq
\end{theorem}

\noindent
{\it Proof.}
From \cite{KBBT95} and the condition that for the matrix CKP
hierarchy $\xi_1$ is a symmetric matrix
it follows that
\beq\label{el6}
\xi_{1, \alpha \beta}=S_{\alpha \beta}-
\sum_{i=1}^N a_i^{\alpha}a_i^{\beta}\zeta (x-x_i),
\eeq
where $S_{\alpha \beta}$ is a symmetric constant matrix, i.e.
the reduction to CKP means that 
$b_i^{\alpha}=a_i^{\alpha}$.
Therefore,
\beq\label{el7}
u_{1, \alpha \beta}=-\sum_{i}a_i^{\alpha}a_i^{\beta}\wp (x-x_i).
\eeq
From (\ref{kp2}) we see that
\beq\label{el11}
\p_{t_2}(a_i^{\alpha}b_i^{\beta})
\Bigr |_{b_{i}^{\alpha}=a_{i}^{\alpha}}=
2\sum_{k\neq i} (a_i^{\alpha}a_k^{\beta}-a_i^{\beta}a_k^{\alpha})
(a_{i}^{\gamma}a_{k}^{\gamma})\wp (x_i-x_k),
\eeq
so $\p_{t_2}(a_i^{\alpha}b_i^{\beta})
\Bigr |_{b_{i}^{\alpha}=a_{i}^{\alpha}}$ is
an antisymmetric matrix.

Now, $w=-\p_{t_2}\xi_1$
must be an antisymmetric matrix function of $x$. Therefore,
the symmetric part of this matrix vanishes, which 
implies that $\p_{t_2}x_i =0$
(at $t_2=0$). This means that the function $w$ 
has first order poles at $x=x_i$.
Using (\ref{el11}), we find its explicit form:
\beq\label{e20}
w_{\alpha \beta}=2\sum_{i}
\sum_{k\neq i} (a_i^{\alpha}a_k^{\beta}-a_i^{\beta}a_k^{\alpha})
(a_i^{\gamma}a_k^{\gamma})\wp (x_i-x_k)\zeta (x-x_i).
\eeq
Since the coefficient in front of the 
$\zeta$-function is antisymmetric with respect to permutation
$i\leftrightarrow k$, $w_{\alpha \beta}$ is indeed 
a double-periodic matrix function of $x$.

After the substitution of $\Psi$ (\ref{el5}) with 
$\rho_{i, \alpha \beta}$ as in (\ref{el5a})
and $u_1$ (\ref{el7}) into the linear problem
\beq\label{el8}
\begin{array}{c}
\p_{t_3}\Psi -\p_x^3\Psi -3u_1\p_x\Psi -\frac{3}{2}\,u_1' \Psi
-\frac{3}{2}\, w\Psi =0,
\end{array}
\eeq
or
\beq\label{el8a}
\begin{array}{c}
\p_{t_3}\Psi_{\alpha \beta} -{\Psi{'''}}_{\!\!\!\!\! \alpha \beta}
-3(u_1)_{\alpha \gamma}\Psi '_{\gamma \beta} -
\frac{3}{2}\,(u_1')_{\alpha \gamma} \Psi_{\gamma \beta}
-\frac{3}{2}\, w_{\alpha \gamma}\Psi_{\gamma \beta} =0,
\end{array}
\eeq
one can see that the left hand side has apparent poles 
at $x=x_i$ up to fourth order. Explicitly, we write (\ref{el8}) 
in the form 
$$
\underbrace{\sum_i \p_{t_3}
(a_i^{\alpha}c_i^{\beta})\Phi (x-x_i)}_{\mbox{order} \,\, 1}-
\underbrace{\sum_i a_i^{\alpha}c_i^{\beta}\dot x_i 
\Phi ' (x-x_i)}_{\mbox{order} \,\, 2}
$$
$$
=3z^2 \underbrace{\sum_i a_i^{\alpha}
c_i^{\beta}\Phi ' (x-x_i)}_{\mbox{order} \,\, 2}
+3z \underbrace{\sum_i a_i^{\alpha}c_i^{\beta}\Phi '' 
(x-x_i)}_{\mbox{order} \,\, 3}+
\underbrace{\sum_i a_i^{\alpha}c_i^{\beta}\Phi '''
 (x-x_i)}_{\mbox{order} \,\, 4}
$$
$$
-3z\underbrace{\sum_i a_i^{\alpha}a_i^{\gamma}
\wp (x-x_i)\sum_j 
a_j^{\gamma}c_j^{\beta}\Phi (x-x_j)}_{\mbox{orders} \,\, 3,\, 2,\, 1}
-3\underbrace{\sum_i a_i^{\alpha}a_i^{\gamma}
\wp (x-x_i)\sum_j a_j^{\gamma}c_j^{\beta}
\Phi ' (x-x_j)}_{\mbox{orders} \,\, 4,\, 2, \, 1}
$$
$$
-\frac{3}{2}
\underbrace{\sum_i a_i^{\alpha}a_i^{\gamma}
\wp ' (x-x_i)\sum_j a_j^{\gamma}
c_j^{\beta}\Phi (x-x_j)}_{\mbox{orders} \,\, 4,\, 3, \, 2, \, 1}
$$
$$
+3\underbrace{\sum_i 
\sum_{k\neq i} (a_i^{\alpha}a_k^{\gamma}-a_i^{\gamma}a_k^{\alpha})
(a_i^{\nu}a_k^{\nu})\wp (x_i-x_k)\zeta (x-x_i)\sum_j 
a_j^{\gamma}c_j^{\beta}\Phi (x-x_j)}_{\mbox{orders} \,\, 2, \, 1}.
$$
Here, under each term, it is indicated which poles it contains.
Cancellation of the fourth order poles gives the conditions
(\ref{el9}):
$a_i^{\gamma}a_i^{\gamma}=1$,
i.e. the vectors ${\bf a}_i$ all have norms equal to 1.
Cancellation of the third order poles gives the equation
\beq\label{e21}
\sum_{j\neq i}(a_i^{\gamma}a_j^{\gamma})
c_j^{\beta}\Phi (x_i-x_j, \lambda )+zc_i^{\beta}=0.
\eeq
From cancellation of the second order poles
we obtain (using (\ref{e21})) equations (\ref{e22}),
which are equations of motion for the poles.
Finally, from the equations (\ref{e24}) and
\beq\label{e25}
\begin{array}{l}
\displaystyle{
\dot c_i^{\beta}=-3(\alpha_1 z +\alpha_2)c_i^{\beta}-3z\sum_{j\neq i}(a_i^{\gamma}a_j^{\gamma})
c_j^{\beta}\Phi '(x_i-x_j, \lambda )-
\frac{3}{2}\, \wp (\lambda )\! 
\sum_{j\neq i}(a_i^{\gamma}a_j^{\gamma})
c_j^{\beta}\Phi (x_i-x_j, \lambda )
}
\\ \\
\displaystyle{\phantom{aaaa}
+\, 3\sum_{j\neq i}\sum_{k\neq i}
(a_i^{\nu}a_k^{\nu})(a_j^{\gamma}a_k^{\gamma})c_j^{\beta}\wp (x_i-x_k)
\Phi (x_i-x_j, \lambda )} +\frac{3}{2}\, c_i^{\beta}\sum_{j\neq i}
(a_i^{\gamma}a_j^{\gamma})^2 \wp '(x_i-x_j)
\end{array}
\eeq
and 
it follows that simple poles cancel, too. 
Equations (\ref{e24}) are the equation of motion
for components of the vectors ${\bf a}_i$. 
\square

\begin{remark}
It is easy to see that equations of motion (\ref{e24}) 
preserve the constraints (\ref{el9}).
\end{remark}

\begin{remark}\label{remerk:spectral}
Equation (\ref{e21}) 
is the eigenvalue equation
\beq\label{e23}
L{\bf c}=z{\bf c}
\eeq
for the Lax matrix
\beq\label{e23a}
L_{ij}(\lambda )
=-(1-\delta_{ij})(a_i^{\gamma}a_j^{\gamma})\Phi (x_i-x_j, \lambda )
\eeq
depending on the spectral parameter $\lambda$. The equation
\beq\label{spectral}
\det \Bigl (zI -L(\lambda )\Bigr )=0
\eeq
defines the spectral curve.  The left-hand side is a polynomial
of $z$ whose coefficients are elliptic functions of $\lambda$. 
It is a generating function for integrals of motion.
Note that 
$\lambda =O(z^{-1})$ as $z\to \infty$.
\end{remark}

\begin{remark}\label{remerk:M}
Equations (\ref{e25}) are evolution equations for
the vectors ${\bf c}_i$ with components $c_i^{\alpha}$ 
of the form
\beq\label{e26}
\dot {\bf c}=M{\bf c}
\eeq
with some matrix $M$ whose explicit form can be read from
(\ref{e25}). (We do not need it here.)
The compatibility condition with the eigenvalue 
equation (\ref{e23}) is the Lax equation
\beq\label{e27}
\dot L +[L,M]=0.
\eeq
However, a direct verification of this equation 
using the equations of motion
is technically difficult.
\end{remark}

\section{Conclusion}

The main result of this paper is equations of motion 
(\ref{e24}), (\ref{e22}) for 
the data of singular solutions (positions of poles 
and matrix coefficients at them) to the matrix CKP equation
as functions of the first nontrivial time flow $t_3$.
Comparing this with the similar result for the matrix 
KP equation, we can note two characteristic features:

\begin{itemize}
\item The equations of motion are of the first order in $t_3$
(rather than second),
\item The equations of motion follow directly from cancellation
of the second order poles which might appear after substitution
of the double-Bloch ansatz into the linear problem (\ref{el8}).
This is in contrast to what happens in the KP case, where 
cancellation of poles in the linear problem does not lead to
the equations of motion directly but first
leads to an
overdetermined system of linear equations for components 
of the vectors ${\bf c}_i$. Its compatibility is equivalent to
the Lax equation (\ref{e27}), from which the equations of motion
can be derived.
\end{itemize}

\noindent
An interesting problem for the future is to generalize this result
to the whole hierarchy, similar to what was done in 
\cite{PZ21spin} for the matrix KP hierarchy.

Finally, we note that the general result (\ref{e24}), (\ref{e22}) 
can be applied to the degenerate cases, too, when one or both periods
tend to infinity. In these cases one should simply substitute 
the degenerate (trigonometric or rational) expressions for the
Weierstrass functions given by (\ref{trig}), (\ref{rat}).

\section*{Appendix: The Weierstrass and Lam\'e-Her\-mi\-te \\ 
func\-ti\-ons}
\addcontentsline{toc}{section}{Appendix: The Weierstrass 
and Lam\'e-Hermite functions}
\def\theequation{A\arabic{equation}}
\def\theHequation{\theequation}
\setcounter{equation}{0}

Here we present the definition and main properties of the 
special functions used in the main text.

Let $\omega_1$, $\omega_2$ be complex numbers such that 
${\rm Im} (\omega_2/ \omega_1 )>0$.
The Weierstrass $\sigma$-function 
with quasi-periods $2\omega_a$, $a=1,2$,
is defined by the following infinite product over the lattice
$2\omega_1 m_1 +2\omega_2 m_2$, $m_1,m_2\in \ZZ$:
\beq\label{A1}
\sigma (x)=\sigma (x |\, \omega_1 , \omega_2)=
x\prod_{s\neq 0}\Bigl (1-\frac{x}{s}\Bigr )\, 
e^{\frac{x}{s}+\frac{x^2}{2s^2}},
\quad s=2\omega_1 m_1+2\omega_2 m_2 \quad m_1, m_2\in \ZZ .
\eeq 
It is an odd quasiperiodic function with two linearly independent
quasi-periods in the complex plane. 
The expansion around $x=0$ is $\sigma (x)=x+O(x^5)$.
The monodromy properties of the $\sigma$-function 
under shifts by the quasi-periods
are as follows:
\beq\label{A4}
\begin{array}{l}
\sigma (x+2\omega_a )=-e^{2\zeta (\omega_a )(x+\omega_a )}\sigma (x),
\quad a=1,2.
\end{array}
\eeq
Here $\zeta (x)$ is the
Weierstrass $\zeta$-function defined as
\beq\label{A4a}
\zeta (x)=\frac{\sigma '(x)}{\sigma (x)}.
\eeq
As $x\to 0$, $\zeta (x)=x^{-1} +O(x^3)$.
The Weierstrass $\wp$-function is defined as
$\wp (x)=-\zeta '(x)$. 
It is an even double-periodic function with periods $2\omega_1 , 2\omega_2$
and with second order poles at the points 
of the lattice $s=2\omega_1 m_1+2\omega_2 m_2$ with integer $m_1, m_2$.
As $x\to 0$, $\wp (x)=x^{-2}+O(x^2)$.

The Lam\'e-Hermite function is defined as
\beq\label{Phi}
\Phi (x, \lambda )=
\frac{\sigma (x+\lambda )}{\sigma (\lambda )\sigma (x)}\,
e^{-\zeta (\lambda )x}.
\eeq
It has simple poles at $x=0$ and all other points 
of the lattice. 
The expansion of $\Phi (x, \lambda )$ as $x\to 0$ is
$$
\Phi (x, \lambda )=\frac{1}{x}+\alpha_1 x +\alpha_2 x^2 +\ldots \, , 
$$
where $\alpha_1=-\frac{1}{2}\, \wp (\lambda )$, 
$\alpha_2=-\frac{1}{6}\, \wp '(\lambda )$. 
The quasiperiodicity properties of the Lam\'e-Hermite function are:
\beq\label{quasi}
\begin{array}{l}
\Phi (x+2\omega_a , \lambda )=e^{2(\zeta (\omega_a )\lambda - 
\zeta (\lambda )\omega_a )}
\Phi (x, \lambda ), \quad a=1,2.
\end{array}
\eeq
As a function of $\lambda$, $\Phi (x, \lambda )$ is a double-periodic
function with periods $2\omega_1$, $2\omega_2$.

In the degenerate case when one of the periods tends to
$\infty$, the elliptic functions degenerate to trigonometric
or hyperbolic ones. Let us choose the first period $2\omega_1$ 
to be $\pi i/\gamma$, where
$\gamma$ is some complex constant (real for hyperbolic functions and purely imaginary
for trigonometric functions).
The second period tends to infinity.
The Weierstrass functions in this limit become
\beq\label{trig}
\begin{array}{l}
\sigma (x)=\gamma^{-1}e^{-\frac{1}{6}\, \gamma^2x^2}\sinh (\gamma x),
\\ \\
\zeta (x)=\gamma \coth (\gamma x)-\frac{1}{3}\, \gamma^2 x,
\\ \\
\displaystyle{
\wp (x)=
\frac{\gamma^2}{\sinh^2(\gamma x)}+{\textstyle \frac{1}{3}}\gamma^2}.
\end{array}
\eeq
The trigonometric limit of the 
Lam\'e-Hermite function is
\beq\label{trig2}
\Phi (x, \lambda )=\gamma \Bigl (\coth (\gamma x)+
\coth (\gamma \lambda )\Bigr )
e^{-\gamma x \coth (\gamma \lambda )}.
\eeq

Finally, when both periods tend to $\infty$, trigonometric 
functions further degenerate to rational ones. In this case
we have
\beq\label{rat}
\sigma (x)=x, \quad \zeta (x)=1/x, \quad \wp (x)=1/x^2
\eeq
and
\beq\label{rat2}
\Phi (x, \lambda )=\Bigl (\frac{1}{x}+\frac{1}{\lambda}\Bigr )
e^{-\lambda x}.
\eeq

\section*{Acknowledgments}

\addcontentsline{toc}{section}{Acknowledgments}

This work was carried out within the state assignment of NRC 
``Kurchatov institute''. It was started in 2021 
jointly with I. Krichever, who took part in it at 
an early stage. I indebted to him for a lot of illuminating
discussions. I am also grateful to the reviewer 
for pointing out an error in equation (48) 
in the original version of the paper.


\begin{thebibliography}{99}

\addcontentsline{toc}{section}{References}

\bibitem{AMM77}
H. Airault, H.P. McKean, and J. Moser, {\it Rational and 
elliptic solutions of the
Korteweg-De Vries equation and a related many-body problem},
Commun. Pure Appl. Math., {\bf 30} (1977) 95--148.



\bibitem{Krichever78}
I.M. Krichever, {\it Rational solutions of the Kadomtsev-Petviashvili
equation and integrable systems of $N$ particles on a line},
Funct. Anal. Appl. {\bf 12:1} (1978) 59--61.


\bibitem{Krichever80} I.M. Krichever, {\it Elliptic solutions of the Kadomtsev-Petviashvili
equation and integrable systems of particles}, Funk. Anal. i Ego Pril. {\bf 14:4} (1980) 45--54
(in Russian); English translation:
Functional Analysis and Its Applications {\bf 14:4} (1980) 282–-290.

\bibitem{CC77} D.V. Chudnovsky and G.V. Chudnovsky, {\it Pole expansions of non-linear
partial differential equations}, Nuovo Cimento {\bf 40B} (1977) 339--350.


\bibitem{Calogero71}
F. Calogero, {\it Solution of the one-dimensional
$N$-body problems with quadratic
and/or inversely quadratic pair potentials}, J. Math. Phys.
{\bf 12} (1971) 419--436.

\bibitem{Calogero75} F. Calogero, {\it Exactly solvable one-dimensional many-body
systems}, Lett. Nuovo Cimento {\bf 13} (1975) 411--415.

\bibitem{Moser75}
J. Moser, {\it Three integrable Hamiltonian systems connected with isospectral
deformations}, Adv. Math. {\bf 16} (1975) 197--220.

\bibitem{OP81} M.A. Olshanetsky and A.M. Perelomov, {\it Classical integrable
finite-dimensional systems related to Lie algebras}, Phys. Rep. {\bf 71} (1981) 313--400.

\bibitem{Shiota94} T. Shiota, 
{\it Calogero-Moser hierarchy and KP hierarchy},
J. Math. Phys. {\bf 35} (1994) 5844-5849.

\bibitem{KZ95} I. Krichever, A. Zabrodin,
{\it Spin generalization of the Ruijsenaars-Schneider model, 
non-abelian 2D
Toda chain and representations of Sklyanin algebra},
Uspekhi Math. Nauk, {\bf 50} (1995) 3--56.

\bibitem{Haine07} L. Haine, {\it KP 
trigonometric solitons and an adelic flag
manifold}, SIGMA {\bf 3} (2007) 015.

\bibitem{Z19a} A. Zabrodin, {\it KP hierarchy and trigonometric 
Calogero-Moser
hierarchy}, Journal of Mathematical Physics {\bf 61} 
(2020) 043502, arXiv:1906.09846.


\bibitem{Z19} A. Zabrodin, {\it Elliptic solutions to integrable nonlinear
equations and many-body systems}, 
Journal of Geometry and Physics 
{\bf 146} (2019) 103506, arXiv:1905.11383.

\bibitem{PZ21} V. Prokofev and A. Zabrodin, 
{\it Elliptic solutions to the KP hierarchy
and elliptic Calogero-Moser model}, 
Journal of Physics A: Math. Theor., {\bf 54} (2021) 305202.

\bibitem{PZ21a}
V. Prokofev and A. Zabrodin, 
{\it Elliptic solutions to Toda lattice hierarchy and elliptic
Ruijsenaars-Schneider model}, Teor. Mat. Fys.,
{\bf 208} (2021), No. 2, 282–-309
(English translation: Theoretical and Mathematical Physics,
{\bf 208} (2021) 1093--1115.

\bibitem{PZ23}
V. Prokofev and A. Zabrodin, 
{\it Elliptic solutions of the Toda lattice with constraint 
of type B and deformed Ruijsenaars-Schneider system}, 
Mathematical Physics, Analysis and Geometry {\bf 26:20} (2023).








\bibitem{DJKM81a} E. Date, M. Jimbo, M. Kashiwara and 
T. Miwa, {\it Transformation groups
for soliton equations III}, J. Phys. Soc. Japan {\bf 50} (1981) 3806--3812.

\bibitem{KL93} V. Kac and J. van de Leur, {\it The 
$n$-component KP hierarchy and 
representation theory}, in: A.S. Fokas, V.E. Zakharov (Eds.), Important Developments
in Soliton Theory, Springer-Verlag, Berlin, Heidelberg, 1993.


\bibitem{TT07} K. Takasaki and T. Takebe, {\it Universal
Whitham hierarchy, dispersionless Hirota equations
and multicomponent KP hierarchy}, Physica D {\bf 235} (2007) 109--125.

\bibitem{Teo11} L.-P. Teo, {\it The multicomponent KP hierarchy: 
differential Fay identities and Lax
equations}, J. Phys. A: Math. Theor. {\bf 44} (2011) 225201.

\bibitem{KBBT95}
I. Krichever, O. Babelon, E. Billey and 
M. Talon, {\it Spin generalization of the
Calogero-Moser system and the matrix 
KP equation}, Amer. Math. Soc. Transl. Ser. 2
{\bf 170} (1995) 83--119.

\bibitem{PZ18}
V. Pashkov and A. Zabrodin, {\it Spin generalization of the
Calogero-Moser hierarchy and the matrix KP hierarchy}, J. Phys. A: Math. Theor. {\bf 51} (2018) 215201.

\bibitem{PZ21spin}
V. Prokofev, A. Zabrodin,
{\it Elliptic solutions to matrix KP hierarchy and spin generalization 
of elliptic Calogero-Moser model}, J. Math. Physics,
{\bf 62} (2021) 061502.

\bibitem{GH84} J. Gibbons and T. Hermsen,
{\it A generalization of the Calogero-Moser system},
Physica D {\bf 11} (1984) 337–-348.

\bibitem{DJKM81}
E. Date, M. Jimbo, M. Kashiwara and T. Miwa,
{\it KP hierarchy of orthogonal and symplectic type -- Transformation
groups for soliton equations VI}, J. Phys. Soc. Japan {\bf 50} (1981) 3813--3818.


\bibitem{DM-H09}
A. Dimakis and F. M\"uller-Hoissen, {\it BKP and CKP revisited: 
the odd KP system},
Inverse Problems {\bf 25} (2009) 045001, arXiv:0810.0757.

\bibitem{CW13}
L. Chang and C.-Z. Wu, {\it Tau function of the
CKP hierarchy and non-linearizable Virasoro symmetries},
Nonlinearity {\bf 26} (2013) 2577--2596.

\bibitem{CH14}
J. Cheng and J. He, {\it The ``ghost'' symmetry in the CKP hierarchy},
Journal of Geometry and Physics, {\bf 80} (2014) 49--57.

\bibitem{LOS12}
J. van de Leur, A. Orlov and T. Shiota, {\it
CKP hierarchy, bosonic tau function and bosonization formulae},
SIGMA {\bf 8} (2012) 036,
arXiv:1102.0087

\bibitem{KZ21}
I. Krichever and A. Zabrodin,
{\it Kadomtsev-Petviashvili turning points and CKP hierarchy},
Commun. Math. Phys. {\bf 386} (2021) 1643--1683,
arXiv:2012.04482  












\bibitem{KL23}
V. Kac, J. van de Leur,
{\it Multicomponent KP type hierarchies and their reductions, 
associated to conjugacy classes of Weyl groups of classical Lie algebras},
J. Math. Phys. {\bf 64} (2023) 9,
arXiv:2304.05737.


\bibitem{Z24}
A. Zabrodin, 
{\it Tau-function of the multi-component CKP hierarchy},
Math. Phys., Analysis and Geometry {\bf 27:1} (2024).



















\end{thebibliography}
\end{document}